\documentclass{ws-ijfcs}
\usepackage{enumerate}
\usepackage{url}
\usepackage{multicol}
\usepackage[utf8]{inputenc}
\usepackage[T1]{fontenc}
\usepackage{amssymb}
\usepackage{amsmath}
\usepackage{amsfonts}
\usepackage{xifthen}
\usepackage{graphicx}
\usepackage{enumerate}
\usepackage{tikz}
\usetikzlibrary{petri}
\usepackage{multicol}

\usepackage{complexity}
\usepackage{mathtools}
\usepackage{comment}
\usepackage{mathrsfs}
\usepackage{longtable}
\usepackage{enumitem}
\usepackage{etoolbox}
\usepackage{float}
\usepackage{color}
\usepackage{wasysym}
\usepackage{hyperref}

\newclass{\COMSLIP}{COM\mbox{-}SLIP}
\newclass{\COMSLIPCUP}{COM\mbox{-}SLIP^{\cup}}

\newclass{\NCSPA}{NCSPA}
\newclass{\DPDM}{DPDM}
\newclass{\NPDM}{NPDM}
\newclass{\DCM}{DCM}
\newclass{\eDCM}{eDCM}
\newclass{\eNPDA}{eNPDA}
\newclass{\PDCSPA}{PDCSPA}
\newclass{\DCSPA}{DCSPA}
\newclass{\DPDA}{DPDA}
\newclass{\DPCSA}{DPCSA}
\newclass{\RDPCSA}{RDPCSA}
\newclass{\RDPDA}{RDPDA}
\newclass{\PDA}{PDA}
\newclass{\DCMNE}{DCM_{NE}}
\newclass{\TwoDCM}{2DCM}
\newclass{\NCM}{NCM}
\newclass{\eNCM}{eNCM}
\newclass{\eNQA}{eNQA}
\newclass{\eNSA}{eNSA}
\newclass{\eNPCM}{eNPCM}
\newclass{\eNQCM}{eNQCM}
\newclass{\eNSCM}{eNSCM}
\newclass{\DPCM}{DPCM}
\newclass{\NPCM}{NPCM}
\newclass{\NQCM}{NQCM}
\newclass{\NSCM}{NSCM}
\newclass{\NPDA}{NPDA}
\newclass{\TRE}{TRE}
\newclass{\NFA}{NFA}
\newclass{\DFA}{DFA}
\newclass{\NCA}{NCA}
\newclass{\DCA}{DCA}
\newclass{\DTM}{DTM}
\newclass{\NTM}{NTM}
\newclass{\NTMCM}{NTMCM}
\newclass{\DLOG}{DLOG}
\newclass{\CFG}{CFG}
\newclass{\ETOL}{ET0L}
\newclass{\EDTOL}{EDT0L}
\newclass{\CFP}{CFP}
\newclass{\ORDER}{O}
\newclass{\MATRIX}{M}
\newclass{\BD}{BD}
\newclass{\LB}{LB}
\newclass{\ALL}{ALL}
\newclass{\decLBD}{decLBD}
\newclass{\StLB}{StLB}
\newclass{\SBD}{SBD}
\newclass{\TCA}{TCA}
\newclass{\LL}{{\cal L}}

\newclass{\NSA}{NSA}

\newclass{\DSA}{DSA}

\newclass{\CSA}{CSA}

\newclass{\DPSPACE}{DPSPACE}
\newclass{\PTIME}{PTIME}

\newclass{\WDCSPA}{WDCSPA}
\newclass{\WDCSA}{WDCSA}
\newclass{\WNCSA}{WNCSA}
\newclass{\DCSA}{DCSA}
\newclass{\NCSA}{NCSA}
\newclass{\GSM}{GSM}
\newclass{\PDCSA}{PDCSA}
\newclass{\PNCSA}{PNCSA}
\newclass{\DCSACM}{DCSACM}
\newclass{\NCSACM}{NCSACM}

\newcommand{\rw}{\rotatebox[origin=c]{180}{$\Rsh$}}

\DeclareMathOperator{\push}{push}
\DeclareMathOperator{\pop}{pop}
\DeclareMathOperator{\sense}{sense}
\DeclareMathOperator{\stay}{stay}
\DeclareMathOperator{\down}{down}
\DeclareMathOperator{\up}{up}

\begin{document}

\markboth{Oscar H. Ibarra, Ian McQuillan}
{Generalizations of Checking Stack Automata: Characterizations and 
Hierarchies}

%
\catchline{}{}{}{}{}
%

\title{Generalizations of Checking Stack Automata: Characterizations and Hierarchies
\thanks{\textcopyright 2022. This manuscript version is made available under the CC-BY 4.0 license \url{https://creativecommons.org/licenses/by/4.0/}. Published at {\it International Journal of Foundations of Computer Science}, 32 (5), 481--508 (2021) \url{https://doi.org/10.1142/S0129054121410045}}}

\author{Oscar H. Ibarra}

\address{
Department of Computer Science, University of California, Santa Barbara, CA 93106, USA\\ \email{ibarra@cs.ucsb.edu}}

\author{Ian McQuillan}

\address{
Department of Computer Science, University of Saskatchewan
	Saskatoon, SK S7N 5C9, Canada\\
	\email{mcquillan@cs.usask.ca}
}

\maketitle

\begin{history}
\received{(Day Month Year)}
\accepted{(Day Month Year)}
\comby{(xxxxxxxxxx)}
\end{history}

\begin{abstract}
We examine different generalizations of checking stack automata
by allowing multiple input heads and multiple stacks,
and characterize their computing power in terms of
two-way multi-head finite automata
and space-bounded Turing machines. 
%
For various models, we obtain hierarchies in terms of their computing power.
Our characterizations and hierarchies expand or tighten
some previously known results.  We also discuss some decidability
questions and the space/time complexity of the models. 

\end{abstract}


\section{Introduction} 
\label{sec:intro}

A one-way stack automaton ($\NSA$)
\cite{GGH} is
a generalization of a one-way pushdown automaton that can enter
its stack in a two-way read-only mode, but can only push or pop
when it is at the top of the stack.
A one-way checking stack automaton \cite{CheckingStack} is a restriction of stack automata
that cannot pop its stack; and it starts by writing to the stack with transitions that either push to the stack or leave the stack unchanged (the {\em write phase}), and then it can read from the inside of the stack (the {\em read phase}), but once in the read phase, it can no longer push to the stack.  
The one-way deterministic and
nondeterministic versions of checking stack automata are denoted by 
$\DCSA$ and $\NCSA$, respectively.
This model is quite powerful, and can accept non-semilinear languages, and $\NCSA$ can accept $\NP$-complete languages
\cite{ShamirBeeri}, although emptiness is still decidable \cite{GGH2}.
Various space complexity measures on $\NCSA$s and their languages were recently studied \cite{DLT2020checking}.

The models with two-way input are called
$2\DCSA$ and $2\NCSA$.  These are
generalized further to models with $k$ checking stacks:
$k$-stack $2\DCSA$ and $k$-stack $2\NCSA$, and with multiple input read heads: $r$-head $k$-stack $2\DCSA$ and
$r$-head $k$-stack $2\NCSA$.
Deterministic multi-head multi-stack $2\DCSA$s were first defined by Vogel and Wagner \cite{VogelWagner}, who showed that
they accept exactly the languages that can be accepted by deterministic Turing machines in $\log$ space.
This class
of automata has an undecidable membership problem for nondeterministic machines \cite{DLT2017TCS}. 
Thus, $r$-head, $k$-stack $2\DCSA$s may be useful in showing other decidable properties.

Here, we compare multi-head multi-stack $2\NCSA$s and $2\DCSA$s to space-bounded Turing machines, and to multi-head
two-way $\NFA$s and $\DFA$s. Also, we examine the space complexity of these stack automata in comparison to other types of machines. Two space-constrained restrictions of $2\NCSA$s are useful:
an $r$-head $k$-stack $2 \NCSA$ is $S(n)$-bounded if every word of length $n$ accepted has an accepting computation where every stack is of length at most $O(S(n))$; and an automaton is $S(n)$-limited if one stack is unbounded but all others are $S(n)$-bounded.
Before summarizing the results, we give some notation to describe the
machine models used.
By convention, if in any of the devices that we have introduced,
$k$ is equal to $1$, we will usually drop the prefix ``$k$-'', and similarly with $r$ when $r$ is equal to $1$. For example, 
$2\NCSA$ will mean $1$-head 1-stack $2\NCSA$, etc.
As usual, one-way deterministic
(nondeterministic) finite automata are denoted by $\DFA$ ($\NFA$), and
the two-way versions are denoted by $2\DFA$ ($2\NFA$), and
the $r$-head varieties are denoted by $r$-head $2\DFA$ ($r$-head $2\NFA$).
One-way deterministic (nondeterministic) pushdown automata are denoted by $\DPDA$ ($\NPDA$), and two-way versions
are denoted by $2\DPDA$ ($\NPDA$).
Deterministic (nondeterministic) Turing machines are denoted by 
$\DTM$ (respectively $\NTM$). An $S(n)$-space-bounded $\DTM$ ($\NTM$) has the
usual meaning \cite{HU}.  
In particular, $\DLOG$ ($\NLOG$) will denote languages accepted in deterministic
(nondeterministic) $\log$ space, $\P$ ($\NP$) will denote languages accepted in deterministic (nondeterministic) polynomial time, and
$\PSPACE$  will denote those accepted in polynomial-space (which coincides for $\DTM$s and $\NTM$s \cite{HU}).




We provide an alternate proof of \cite{VogelWagner} that the following models accept the same languages:
\begin{multicols}{2}
\begin{itemize}
\item
multi-head multi-stack $2\DCSA$,
\item multi-stack $n$-limited $2\DCSA$,
\item multi-head $2\DFA$,
\item $\log n$ space-bounded $\DTM$.
\end{itemize}
\end{multicols}
\noindent In the process of proving the above equivalences, we provide new trade-offs
between the number of heads/stacks of the multi-head multi-stack $2\DCSA$ and
the number of heads of the multi-head $2\DFA$ when converting one model
to the other.
This result is generalized to a new model that is nondeterministic 
but the $2\NCSA$ is restricted so that it operates deterministically
until every stack has entered the reading phase, and thereafter
can work nondeterministically; the languages are similarly all in $\P$. This model is possibly useful as a method of developing polynomial time algorithms as this multi-head multi-stack $2\NCSA$ restriction
does not have an explicit time or space bound.
Also, the following hierarchies are shown
for $r,k \ge 1$: $r$-head $k$-stack $2\DCSA$ is weaker than
$(r+2)$-head $k$-stack $2\DCSA$, and
$r$-head $k$-stack $2\DCSA$ is weaker than
$r$-head $(k+2)$-stack $2\DCSA$.
This improves a result from \cite{KingWrathall} that
$2\DCSA \subsetneq \DLOG$.
We also prove analogous results for when the multi-stack $2\DCSA$ is augmented
with a pushdown stack which can only be used after all the checking stacks
are done writing. In particular,
it is shown that this model is equivalent to a $2\DPDA$ augmented with a $\log n$ space-bounded
worktape which we shall refer to as a $\log n$ space-bounded auxiliary $2\DPDA$ ($2\NPDA$). This model
was introduced and studied by Cook in \cite{Cook71}, who showed that the
nondeterministic and deterministic versions are equivalent.

It is shown that for every function $S(n)$,
$(k+1)$-stack $S(n)$-limited $2\NCSA$ languages are equal to $n S^k(n)$ space-bounded $\NTM$ languages.
Furthermore, this allows us to conclude that
$(k+1)$-stacks are more powerful than $k$-stacks for some functions $S(n)$ by using separations results known regarding space complexity of Turing machines.  Hence,
multi-stack $n$-limited $2\NCSA$ languages are equal to $\PSPACE$, as are 
2-stack polynomial-space-limited $2\NCSA$ languages.
Also, $S(n)$-bounded $2 \NCSA$s are analyzed and compared to other types of machines.
In addition, machines with one global read phase are analyzed.

This paper is an extended version of \cite{DLT2018checking}, but with many previously omitted proofs included, and several new results including all of Section \ref{sec:alllinear} and everything starting with Section \ref{sec:phasedmachines}.

\section{Preliminaries}
\label{sec:prelims}

This paper assumes a working knowledge of automata and formal language theory, including finite automata and Turing machines.
Please see \cite{HU} for basic models, notation for these models, and results. 
An {\em alphabet} is a finite set of symbols. Given an alphabet $\Sigma$, $\Sigma^*$ is the set of all words over $\Sigma$ which includes the empty word $\lambda$. A {\em language} $L$ over $\Sigma$ is any subset of $\Sigma^*$.

Next, we define the main machines of interest. It is most natural to define stack automata, and then checking stack
via a restriction. Fixed across all machines are the left and right input end-markers $\rhd,\lhd$, the bottom and top
stack markers $Z_b, Z_t$, and the stack read/write head denoted by $\rw$.

A {\em two-way $r$-head $k$-stack nondeterministic stack automaton} ($r$-head $k$-stack $2\NSA$) is a tuple $M = (Q,\Sigma,\Gamma,\delta,q_0,F)$, where $Q$ is a finite set of states, $q_0 \in Q$ is the initial state, $F \subseteq Q$ is the set of final states, $\Sigma$ is the finite input alphabet, 
and $\Gamma$ is the finite stack alphabet.
Let $\Sigma_{+} = \Sigma \cup \{\rhd,\lhd\}, \Gamma_b = \Gamma \cup \{Z_b\}, \Gamma_t = \Gamma \cup \{Z_t\}, \Gamma_{+} = \Gamma \cup \{Z_b, Z_t\}, I  = \{\push(y), \pop, \stay,\down,\up \mid y \in \Gamma\}$ (the stack instructions, the first three being write instructions that can apply at the top of the stack, and the latter three being read instructions that can apply when reading inside the stack.
Then, $\delta$ is the transition function, which is a partial function from
$Q \times \overbrace{\Sigma_{+} \times \cdots \times \Sigma_{+}}^r \times \overbrace{\Gamma_{+} \times \cdots \times \Gamma_{+}}^k$ to the powerset of
$Q \times \overbrace{I \times \cdots \times I}^k \times \overbrace{\{-1,0, +1\} \times \cdots \times \{-1,0, +1\}}^r$, with each such transition denoted by
$$\delta(q, a_1, \ldots, a_r, b_1, \ldots, b_k) \rightarrow (q', \iota_1, \ldots, \iota_k,\gamma_1, \ldots, \gamma_r),$$
whereby for each $i$, $1 \leq i \leq r$, and each $j$, $1 \leq j \leq k$,
\begin{itemize}
\item $a_i = \rhd$ (respectively $\lhd$) implies $\gamma_i \neq -1$ ($\gamma_i \neq +1$),
\item $\iota_j = \down$ (respectively $\iota_j = \up$) implies $b_j \neq Z_b$ ($b_j \neq Z_t$),
\item $\iota_j = \push(y)$ implies $b_j \neq Z_t$, 
\item $\iota_j = \pop$ implies $b_j \in \Gamma$.
\end{itemize}
Such a machine is deterministic if there is at most one transition from
each $q \in Q, a_1, \ldots, a_r \in \Sigma_{+}, b_1, \ldots, b_k \in \Gamma_+$.

An instantaneous description (ID) is a tuple
$(q, \rhd w \lhd, \alpha_1, \ldots, \alpha_r, x_1, \ldots, x_k)$, where $q \in Q$ is the current state, $w \in \Sigma^*$ is the input
word (between end-markers $\rhd$ and $\lhd$ in the ID), $\alpha_i \in\mathbb{N}_0$ with $0 \leq \alpha_i \leq |w|+1$ is the current position of tape head $i$ (this can be thought of as $0$ scanning $\rhd$, $|w|+1$ scanning $\lhd$) for $1 \leq i \leq r$, and 
$x_j \in Z_b \Gamma^* \rw \Gamma^* Z_t \cup Z_b \Gamma^* Z_t \rw$ is the stack contents of stack $j$ (with the stack head $\rw$ reading the character before it) for $1 \leq j \leq k$.
Define $$(q, \rhd w \lhd, \alpha_1, \ldots, \alpha_r, x_1, \ldots, x_k) \vdash_M (q', \rhd w \lhd, \alpha_1', \ldots, \alpha_r', x_1', \ldots, x_k')$$ if there exists a transition
$\delta(q, a_1, \ldots, a_r, b_1, \ldots, b_k) \rightarrow (q', \iota_1, \ldots, \iota_k,\gamma_1, \ldots, \gamma_r),$
where for $1 \leq i \leq r$, $a_i$ is the character at position $\alpha_i+1$ of $\rhd w \lhd$, and $\alpha_i' = \alpha_i + \gamma_i$, and for $1 \leq j \leq k$, 
$x_j = y_j b_j \rw z_j$, and
\begin{itemize}
\item $\iota_j = \push(y)$ implies $z_j = Z_t, x_j' = y_j b_jy Z_t$,
\item $\iota_j = \pop$ implies $z_j = Z_t, x_j'= y_j Z_t$,
\item $\iota_j = \stay$ implies $x_j' = x_j$
\item $\iota_j = \down$ implies $x_j' = y_j \rw b_j z_j$,
\item $\iota_j = \up$ implies $z_j = d_j z_j', d_j \in \Gamma_t, x_j' = y_j b_j d_j \rw z_j'$.
\end{itemize}
As usual, $\vdash_M^*$ is the reflexive and transitive closure of $\vdash_M$, and an accepting computation on $w \in \Sigma^*$ is a sequence
$$(q_0,\rhd w \lhd, 1, \ldots, 1, Z_b \rw Z_t, \ldots, Z_b\rw Z_t) \vdash_M^* (q_f, \rhd w \lhd, \alpha_1, \ldots, \alpha_r, x_1, \ldots, x_k),$$
where $q_f \in F$, and the language accepted by $M$, denoted by $L(M)$, is the set of all strings $w \in \Sigma^*$ such that there is an accepting computation on $w$.

Furthermore, $M$ is a non-erasing stack automaton if there are no pop transitions. And, a non-erasing stack automaton
is a checking stack automaton if, in every accepting computation, and each stack $1$ through $k$, the sequence
of transitions only applies instructions in $\{\push(y), \stay \mid y \in \Gamma\}$ (called the write phase), followed
by $\{\stay,\down, \up\}$ (called the read phase).
As mentioned in Section \ref{sec:intro}, we denote the class of two-way $r$-head $k$-stack nondeterministic checking stack automata by $r$-head $k$-stack $2\NCSA$, and we replace N by D for deterministic machines. We use the phrase multi-head to mean $r$-head for some $r\geq 1$, and multi-stack to mean $k$-stack for some $k \geq 1$. When there are $0$ stacks,
we replace $\NCSA$ and $\DCSA$ with $\NFA$ and $\DFA$ as they clearly match nondeterministic and deterministic finite automata.

\smallskip
\begin{example}
\label{division}
A classical example of a $\DCSA$ language is $\{a^i b^m \mid i \mbox{~divides~} m\}$.
Indeed a $\DCSA$ machine can read $a^i$ while pushing $a^i$ onto the checking stack. As it reads input $b^m$, it alternates between making right-to-left and left-to-right sweeps on the checking stack reading $a^i$, verifying that it hits the end of the input at the same time as it hits an end of the stack.
\end{example}

\begin{example}
Next, consider the language $\{a^i b^j c^m \mid i\cdot j = m\}$. This can be accepted by a 2-stack $\DCSA$ $M$. Here, $M$ copies $a^i$ from the input onto the first stack and then input $b^j$ onto the second stack. As it reads the input $c^m$, it makes sweeps on the first stack as in Example \ref{division}, but it moves left one cell on the second stack for every sweep. It then verifies that all three heads (input and two stacks) reach their ends simultaneously. 
\end{example}

\section{Hierarchies of Multi-Head Multi-Stack $2\DCSA$ and $2\NCSA$}


\subsection{The Deterministic Case}

This section will consider multi-head, multi-stack $2\DCSA$s.

Obviously, a $k$-stack $\DCSA$ ($\NCSA$) (which has
a one-way input) is a special case of
a $k$-stack $2\DCSA$ ($2\NCSA$).  It is also immediate 
that a $k$-stack $2\DCSA$ ($2\NCSA$) $M$ can be
simulated by a one-way $(k+1)$-stack $\DCSA$ ($\NCSA$) $M'$
which first copies its one-way input into one 
stack which can scan this copied input in a two-way fashion to simulate the scanning of the two-way input head of $M$.

We will expand on the following results from \cite{EngelfrietCheckingStack}:
\begin{proposition} \cite{EngelfrietCheckingStack} \label{Engel}
$~~~$
\begin{enumerate}
\item
A language $L$ is accepted by a $2\DCSA$ if and only if
it can be accepted by a 2-head $2\DFA$.
\item
Let $r \ge 1$. If a language $L$ is accepted by an $(r+1)$-head $2\DFA$,
then it can be accepted by an  $r$-head $2\DCSA$.
\item
Let $r \ge 1$. If a language $L$ is accepted by an $r$-head $2\DCSA$,
then it can be accepted by a $(3r +1)$-head $2\DFA$.
\end{enumerate}
\end{proposition}

To help, we create a new type of machine model called an {\em $r$-head $k$-stack $2\NCSA$ with sensing}. Such a machine
has an additional test called $\sense(i,j)$, for $1 \leq i < j \leq r$ that can test whether input heads $i$ and $j$ are 
pointing at the same cell. So in this case, transitions can: read the current state, read the symbol under the $r$ input heads,
read the current cell from each of the stacks, but it can also tell whether any two
given input heads are pointing at the same cell (in response, they can switch states, move the input heads, and can $\push$ to the stacks or read from inside the stacks as normal).

Such a machine $M$ is in 
{\em normal form} if it uses the states to keep track of which stacks are in their writing phase or reading phase,
every transition (move) applied during the computation writes a symbol to each stack
that has not yet entered its reading phase, 
and it remembers which symbol was last written to each stack using the state.
\begin{lemma}\label{normalform}
Given an $r$-head $k$-stack $2\NCSA$ ($2\DCSA$) with sensing $M$, there exists an $r$-head $k$-stack $2\NCSA$ ($2\DCSA$) with sensing $M'$ in normal form such that $L(M) = L(M')$.
\end{lemma}
\begin{proof}
Clearly, by adding states, we may assume that $M$ can keep track of which stacks are 
in their writing  phase or reading phase.  We then construct from $M$ a new $r$-head $k$-stack
$2\DCSA$ $M'$ which uses a new  dummy stack symbol $\#$.  During each transition, 
for each stack
that has not yet entered its reading phase, if $M$ does not write on the stack, $M'$ writes $\#$. When a stack encounters $\#$ during its reading phase, $M'$ skips all $\#$ symbols until it reads a symbol of $\Gamma$. 
It can also remember the last pushed symbol in the state.
\end{proof}

To simplify the upcoming constructions, we first generalize 
Proposition \ref{Engel} Part 3 to
multi-stack
$2\DCSA$ with multiple two-way input heads and sensing.
The construction illustrates the power of multiple input heads, as
a stack can be eliminated.
\begin{lemma} \label{lem-red}
Let $r \ge 1, k \ge 1$.
For every $r$-head $k$-stack $2\DCSA$ with sensing $M$ in normal form, there is
a $(3r)$-head $(k-1)$-stack $2\DCSA$ with sensing $M'$ in normal form such
that $L(M') = L(M)$.
\end{lemma}
\begin{proof}
Let $M$ have $l$ states and let the stacks be called 
$S_1, \ldots, S_k$. 
Construct a $(3r)$-head $(k-1)$-stack $2\DCSA$ with sensing $M'$ 
that has two steps on input $w$:

Initially, $M'$ determines which of the $k$  stacks of $M$ 
first enters its reading phase
if it has a finite writing phase. To do this,
$M'$ will use $r$ heads, $H_1, \ldots, H_r$, to simulate the writing phase of $M$ without simulating the writing on
the stacks. The way that $M'$ does this is using the following:
no stack of $M$ enters the reading phase
if and only if $M$ has made more than $m=l(|w|+2)^r$ moves without
a stack entering its reading phase, otherwise there will be two IDs hit with the same state and input head positions, and so determinism causes an infinite loop (the `$2$' is due to the
end-markers, and also the position of the read heads implies the result of any sense test). Thus, as $M'$ is reading the input to simulate $M$ using $H_1, \ldots, H_r$, $M'$ uses another set of $r$ heads, $I_1, \ldots, I_r$,
to count the number of moves of $M$ up to $m$.  
To do this counting, $M'$ can use each head to count to $|w|+2$, and can therefore simulate an $r$-digit base $|w|+2$ number
that can count to $(|w|+2)^r$ by using the $r^{\mbox{th}}$ head as the least significant digit and the first head as the most significant digit. Finally, it can use the states of $M'$ to count to $l$ for each change in the simulated number, thereby allowing to count to $m$.
If no stack of $M$ enters a
reading phase, $M'$ halts and rejects.  
If $M$ has a finite writing phase, $M'$ can
determine the stack, $S_i$ say, which enters the reading phase first by simulating a transition that reads down in the stack within the first $m$ moves.

For the second step, $M'$ uses (the same as above) heads $H_1, \ldots, H_r$ 
to simulate the $r$ heads of $M$,
as well as several additional two-way heads to simulate 
the computation of $M$ on stack $S_i$.  In the simulation of $M$,
$M'$ uses its $k-1$ stacks to simulate stacks 
$S_1, \ldots, S_{i-1} , S_{i+1}, \ldots, S_k$, but
simulates the operation of $S_i$ using Lemma \ref{normalform} and
a clever technique in
\cite{EngelfrietCheckingStack}.
We describe the construction. The simulation of the
writing phase of $S_i$ is straightforward.   But since
$M'$ cannot record the stack contents of $S_i$,
$M'$ needs to be able to recover the symbols
written on $S_i$ as they are needed during $S_i$'s
reading phase.  To accomplish this, heads $I_1, \ldots, I_r$ and a state component of $M$
are used to represent the current configuration of $M$
at the time when it pushed the $d^{\mbox{th}}$ symbol on $S_i$, for some $d$. These $r$ heads point at the position of the $r$ heads of $M$ when the symbol at position $d$ of $S_i$ was written, the state component $q$ stores the state of $M$ at this point and the top of all the stacks at this point.  If $M'$ needs to recover the symbol on the
next cell $d+1$ to the right of cell $d$,
$M'$ simulates one transition from $I_1, \ldots, I_r$ and $q$ to determine the
symbol written on cell $d+1$ (in $M$, a symbol is written on $S_i$ in each transition of its writing phase due to the normal form). However, if $M'$ needs
the symbol on cell $d-1$ (to the left
of the current cell  $d$), this is more complicated. In this case, $M'$ uses another set of $r$ heads, $J_1, \ldots, J_r$, to
simulate the writing phase of $M$ from the beginning while also remembering the last transition applied in the finite control. After simulating
each transition of $M$, $M'$
checks that the configuration of the $r$ heads $J_1,\ldots, J_r$ are in the same positions as $I_1, \ldots, I_r$ respectively. 
If not, it continues. But if so, it has
recovered the final transition applied which can be ``undone'' to recover the symbol
written in position $d-1$ of the stack (since one symbol is written in each step of the writing phase), and to modify $I_1, \ldots, I_r$ appropriately. 

Thus, $M'$ needs a total of $3r$ heads to simulate $M$,
and $M'$ has only $k-1$ stacks. Clearly, $L(M') = L(M)$.
Afterwards, $M'$ can be again placed in normal form by applying Lemma \ref{normalform}.
\end{proof}

The construction in Lemma \ref{lem-red} can be applied twice to construct
a new $2\DCSA$ with sensing
$M''$ with  $3^2r$ heads and $k-2$ stacks, etc.
By iterating, we have:
\begin{proposition} \label{prop17}
Let $r \geq 1, k \geq 1$. For every $r$-head $k$-stack $2\DCSA$ with sensing $M$, there is
a $(3^kr)$-head $2\DFA$ with sensing $M'$ such that $L(M') = L(M)$.
\end{proposition}

Further, the ability to sense can be easily removed from $2\DFA$s.
\begin{proposition}
\label{removesensing}
Let $r \geq 2$. For every $r$-head $2\NFA$ (resp.\ $2\DFA$) with sensing $M$, there is an
$(r+1)$-head $2\NFA$ (resp.\ $2\DFA$) without sensing $M'$ such that $L(M) = L(M')$.
\end{proposition}
\begin{proof}
Call the $r$ heads $H_1, \ldots, H_r$. Another head called the {\em sensing head} is created.
The sensing head
can be used to check if any two heads of $H_1, \ldots, H_r$ are on the same input
position by moving the two heads towards the left end-marker while moving the
sensing head to the right from the left end-marker, allowing to test for equality
plus allowing recovery of their original position. 
\end{proof}

Combining together the previous two propositions, we obtain:
\begin{corollary} \label{head-stack-dfa}
\label{alltoDFA}
Let $r \geq 1, k \geq 1$. For every $r$-head $k$-stack $2\DCSA$ (with or without sensing) $M$, there is
a $(3^kr + 1)$-head $2\DFA$ $M'$ such that $L(M') = L(M)$.
\end{corollary}

It is worth pointing out that we did not need to create $2\DCSA$s with sensing, but  instead could have built the sensing
head from the proof of Proposition \ref{removesensing} directly into the proof of Lemma \ref{lem-red}; but this would have unnecessarily increased the number of heads when iterating
Lemma \ref{lem-red}.

Later, we will give an analogue of Proposition \ref{prop17} for a restricted version of
multi-head multi-stack $2\NCSA$.
It follows from Corollary \ref{alltoDFA} (by setting $r = 1$)
that every language that is accepted by a $k$-stack $2\DCSA$
can be accepted by a $(3^k+1)$-head $2\DFA$.
However, we will show that we can improve this result
to $2 \cdot 3^{k-1}+1$ heads.

\begin{lemma} \label{lem-red1}
Let $k \ge 2$. For every $k$-stack $2\DCSA$ $M$, there is
a $2$-head $(k-1)$-stack $2\DCSA$ $M'$
such that $L(M') = L(M)$. 
\end{lemma}
\begin{proof}
We will need the following result concerning
two-way deterministic generalized sequential machines
($2\GSM$s) with input end-markers and accepting states
in \cite{EngelfrietCheckingStack},
the proof of which is a generalization
of an earlier result in \cite{AhoUllman1970}.
The result states:
Let $X$ be an arbitrary storage structure, and $r \ge 1$.
Then the class of languages accepted by $r$-head $2\DFA$s
augmented with storage structure $X$ is closed under inverse $2\GSM$
mappings.

Thus, in particular, the class of languages accepted
by $r$-head $k$-stack $2\DCSA$s is closed under
inverse $2\GSM$ mappings.

Let $\Sigma$ be the input alphabet,  $\Gamma_1, \ldots, \Gamma_k$
be the stack alphabets of stacks $S_1, \ldots, S_k$ respectively of $M$
(which we assume without loss of generality to be disjoint),
and $\Gamma = \Gamma_1 \cup \cdots \cup \Gamma_k$.

First construct a 
$2\GSM$ $T$ which, for every 
input $w \in \Sigma^*$ 
of $M$,  outputs $T(w) = w \# z$,  where $z \in \Gamma^*$ is the string
which represents 
the ``combined'' strings $M$ writes on the $k$ stacks
just before one of the stacks enters the reading phase.
(Note that the stacks may have separate reading phases.)  $T$ does this
by first outputting the input $w$ followed by $\#$ and then simulating
the writing phase of $M$ until one of the
stacks enters the reading phase by outputting the stack contents.  So, e.g., if in a step, $M$ writes $s_i$  on
stack $S_i$  (where $s_i$ is a symbol or $\lambda$), then in the simulated
step, $T$ outputs $s_1 \cdots s_k$. (Note that if none of the stacks
enters the reading phase, $T$ does not halt.) If stack $S_p$ ($1 \le p \le k$)
enters the reading phase first,
$T$ enters an accepting state.

Now construct a 2-head $(k-1)$-stack $2\DCSA$ $M'$ which, when given input $w \# z$,
first simulates the writing phase of $M$ on input $w$ by using one 
head on $w$ and the other input head to check that $z$ is the string
representing the ``combined''  strings written on the $k$ stacks and that
$S_p$ is the first stack
that enters the writing phase.
Then $M'$ scans $z$ and writes into the $k-1$ stacks other than $S_p$ the strings
that would have been written in these stacks.
Then $M'$ continues the simulation by using one input head 
on $w$ and the other head on $z$ to simulate
stack $S_p$. Note that 
when working on $z$, the input head simulating stack $S_p$ skips all 
symbols not in its stack alphabet, $\Gamma_p$.
The other $k-1$ stacks are simulated faithfully.
Then $T^{-1} (L(M')) = L(M)$. The result 
follows from the result from \cite{EngelfrietCheckingStack}.
\end{proof}

This gives a nice characterization of
2-head $2\DCSA$s:
\begin{corollary} \label{prop-red3}
A language $L$ is accepted by a 2-head $2\DCSA$
if and only if $L$ can be accepted by a 2-stack $2\DCSA$.
\end{corollary}
\begin{proof}
Let $M$ be a 2-head $2\DCSA$.  Construct a 2-stack
$2\DCSA$ $M'$ which first copies the input on
one stack. Thus $M'$ has now two copies of the
input, and $M'$ simulates $M$ using the two
copies. 
The converse follows from Lemma \ref{lem-red1}.
\end{proof}

Next, part 1 of the following proposition uses Lemma \ref{lem-red1} to improve  Corollary \ref{alltoDFA} when there is one head,
and part 2 converts multiple heads to stacks.
\begin{proposition} \label{prop-red2}
$~~~$
\begin{enumerate}
\item
Let $k \ge 2$. For every $k$-stack $2\DCSA$ $M$, there is 
a $(2 \cdot 3^{k-1} +1)$-head $2\DFA$ $M'$
such that $L(M') = L(M)$.
\item
Let $k \ge 1$. For every $(k+1)$-head $2\DFA$ $M$, there is
a $k$-stack $2\DCSA$ $M'$ such that
$L(M') = L(M)$.
\end{enumerate}
\end{proposition}
\begin{proof}
For Part 1, we first apply Lemma \ref{lem-red1},
and then apply
Corollary \ref{alltoDFA}.
For Part 2, $M'$ copies the input into the $k$ stacks and then
simulates $M$ using the input and the stacks.
\end{proof}

From Corollary \ref{alltoDFA}, Proposition \ref{prop-red2},
and the known equivalence of multi-head $2\DFA$s and
$\log n$ space-bounded $2\DTM$s \cite{Ibarra-JCSS71}, we obtain an alternate proof of the following from \cite{VogelWagner}:
\begin{corollary} \label{equiv}
The languages accepted by the following coincide:
\begin{multicols}{2}
\begin{itemize}
\item multi-stack $2\DCSA$s,
\item multi-head multi-stack $2\DCSA$s,
\item multi-head $2\DFA$s, 
\item $\log n$ space-bounded $\DTM$s.
\end{itemize}
\end{multicols}
\end{corollary}
Since languages accepted by  $\log n$ space-bounded $\DTM$s are in $\P$,
 it follows that languages accepted by multi-head multi-stack $2\DCSA$s
are in $\P$.

Looking now at the relationship between the $k$-head $2\DCSA$ and the $k$-stack $2\DCSA$ for $k \ge 2$, any language accepted by a $k$-head $2\DCSA$ can be accepted by
a $k$-stack $2\DCSA$ (this is obvious).  By Proposition \ref{prop-red2},
any language
accepted by a $k$-stack  $2\DCSA$ can be accepted by a $(2 \cdot 3^{k-1}+1)$-head $2\DFA$ which, in turn, can be accepted by a $(2 \cdot 3^{k-1})$-head $2\DCSA$ (again, this is obvious).  Hence, we have the following:
\begin{corollary} \label{head-stack}
For $k \ge 2$:
\begin{enumerate}
\item
Any language accepted by a $k$-head $2\DCSA$ can be accepted by a $k$-stack $2\DCSA$.
\item
Any language accepted by a $k$-stack  $2\DCSA$ can be accepted by a $(2 \cdot 3^{k-1})$-head $2\DCSA$.
\end{enumerate}
\end{corollary}
It would be interesting to know if in the second item above, the
$(2 \cdot 3^{k-1})$ can be improved for $k \ge 3$, noting that
for $k = 2$, the `6' can be improved to `2', since as shown
in Corollary \ref{prop-red3}, $2$-head $2\DCSA$ and
$2$-stack $2\DCSA$ are equivalent. Note that for $k=1$, the models
are identical.

The following demonstrates a hierarchy based on read heads and stacks:
\begin{proposition} \label{general}
For $r,k \ge 1$:
\begin{enumerate}
\item
$r$-head $k$-stack $2\DCSA$ is weaker than
$(r+2)$-head $k$-stack $2\DCSA$.
\item
$r$-head $k$-stack $2\DCSA$ is weaker than
$r$-head $(k+2)$-stack $2\DCSA$.
\end{enumerate}
\end{proposition}
\begin{proof}
Consider part 1.
From Corollary \ref{head-stack-dfa},
%
%
there is an $r_1 > r$ such that
every language accepted by an $r$-head $k$-stack $2\DCSA$
can be accepted by an  $r_1$-head $2\DFA$. 
Now there is an infinite hierarchy of multi-head $2\DFA$s
with respect to their recognition power in terms of the number
of heads \cite{Ibarramultihead,Monien1,Monien}.
Hence, there is an $r' > r_1$
such that the class of languages accepted by $r_1$-head
$2\DFA$s is properly contained in the class of languages
accepted by $r'$-head $2\DFA$s. It follows that 
for any $r, k \ge 1$, there is an $r' > r$ such that the class of 
languages accepted by $r$-head $k$-stack $2\DCSA$s is
properly contained in the class of languages accepted
by $r'$-head $k$-stack $2\DCSA$s. We now show that
$r'$ can be reduced to $r+2$.

In \cite{Ibarramultihead}, using a
translational technique, the following result
was shown: Let $X$ be an arbitrary storage structure
(e.g., $k$-checking stacks, a pushdown, or a combination of both).
If for every $r \ge 1$, there is an $r' > r$ such the class of languages
accepted by $r$-head $2\DFA$s augmented with $X$
is properly contained in the class of languages accepted by
$r'$-head $2\DFA$s augmented with $X$, then
$r'$ can be reduced to 2. The proof of this consisted of Lemma 1, Lemma 2,
and Theorem 1 (part a) in \cite{Ibarramultihead}. The proof is
for multi-head $2\DFA$s, but one can easily check that the proof
applies directly to multi-head $2\DFA$s augmented with $X$,
hence, in particular, when $X$ consists of $k$ checking stacks
and/or a pushdown stack (since the proof does not involve
modification of how the storage structure operates).  

Part 2 follows from part 1 since an $(r+2)$-head $k$-stack $2\DCSA$
can be simulated by an $r$-head $(k+2)$-stack $2\DCSA$.
(Just copy the input on 2 new stacks and then simulate
2 input heads on the 2 new stacks.)
\end{proof}

It is an interesting open question whether in Proposition \ref{general},
$r+2$ can be reduced to $r+1$ in part 1 (and, consequently also reduce
$k+2$ to $k+1$ in part 2).
In \cite{Monien1,Monien}, it was shown that $r$-head $2\DFA$s
are weaker than $(r+1)$-head $2\DFA$s, but it is not clear that
the techniques in these papers can be used to resolve this open 
question. However,
for the special case when $r = k = 1$, we can improve Proposition \ref{general}: 
\begin{proposition} \label{stack1-2-4}
For $r \ge 1$:
\begin{enumerate}
\item
$1$-head $1$-stack $2\DCSA$ is weaker than $2$-head $1$-stack $2\DCSA$. 
\item
$1$-head $1$-stack $2\DCSA$ is weaker than $1$-head $2$-stack $2\DCSA$.
\end{enumerate}
\end{proposition}
\begin{proof}
A $1$-head $1$-stack $2\DCSA$ is equivalent
to a $2$-head $2\DFA$ (by Proposition \ref{Engel}, part 1)
which, in turn, is weaker than a $3$-head $2\DFA$ \cite{Monien1,Monien}.
Part 1 then follows, since a $3$-head $2\DFA$ can easily be simulated
by a $2$-head $1$-stack $2\DCSA$.
Part 2 follows from Part 1, since a $2$-head $1$-stack
$2\DCSA$ can trivially be simulated by a $1$-head $2$-stack $2\DCSA$.
\end{proof}

The constructions and proofs for multi-head and multi-stack $2\DCSA$s
can be used to obtain similar results for more general models. For
example, consider the model of 
$r$-head $k$-stack $2\DCSA$ augmented with a pushdown store called $r$-head $k$-stack
$2\DCSPA$. We say such a machine is {\em ordered} if
the machine can only start using the pushdown
after all the checking stacks are done writing.
If a machine is ordered, this essentially means it behaves as if the pushdown is not there until the writing 
on the checking stacks is finished.
Call this model $r$-head $k$-stack ordered $2\DCSPA$.
Because of the restriction on the machine
regarding when it can start using the pushdown
and the known result that $r$-head $2\DPDA$
is weaker than $(r+1)$-head $2\DPDA$ \cite{Ibarramultihead},
one can verify that all the results we have obtained so far
in this section would hold when in the statements of the results,
$2\DFA$ and $2\DCSA$ are replaced by $2\DPDA$ and ordered $2\DCSPA$, respectively.
So, in particular, as in 
Proposition \ref{Engel} part (1), Corollary \ref{head-stack-dfa},
Corollary \ref{prop-red3},
Proposition \ref{prop-red2},
Corollary \ref{head-stack}, and Proposition \ref{general}, we have:
\begin{proposition}
$~~$
\begin{enumerate}
\item
A language is accepted by an ordered $\DCSPA$ if and only if
it can be accepted by a $2$-head $\DPDA$. 
\item
For $r \geq 1, k \geq 1$, any language accepted by an
$r$-head $k$-stack ordered $2\DCSPA$
(with or without sensing) can be accepted by
a $(3^kr + 1)$-head $2\DPDA$.
\item
A language is accepted by a 2-head ordered $2\DCSPA$ if and only if
it can be accepted by a 2-stack ordered $2\DCSPA$.
\item
For $k \ge 2$, any language accepted by a $k$-stack ordered $2\DCSPA$
can be accepted by a
$(2 \cdot 3^{k-1} +1)$-head $2\DPDA$.
\item
For $k \ge 1$, any language accepted by a $(k+1)$-head $2\DPDA$
can be accepted by a $k$-stack ordered $2\DCSPA$.
\item
For $k \ge 2$, any language accepted by a $k$-head ordered $2\DCSPA$
can be accepted by a $k$-stack ordered $2\DCSPA$.
\item
For $r, k \ge 1$,
$r$-head $k$-stack ordered $2\DCSPA$ is weaker than $(r+2)$-head $k$-stack ordered $2\DCSPA$. 
\item
For $r, k \ge 1$,
$r$-head $k$-stack ordered $2\DCSPA$ is weaker than $r$-head $(k+2)$-stack ordered $2\DCSPA$.
\end{enumerate}
\end{proposition}

The analogue of Proposition \ref{stack1-2-4}
with $2\DCSA$ replaced by ordered $2\DCSPA$ also holds
--- the proof is analogous, using the fact that
$r$-head $2\DPDA$ is weaker than $(r+1)$-head $2\DPDA$
\cite{Ibarramultihead}.

Similarly, Proposition \ref{equiv} with
$\log n$ space-bounded $\DTM$s replaced 
by $2\DPDA$s augmented with $\log n$ space-bounded~read/write tape,
would also hold.  This latter model, 
$\log n$ space-bounded auxiliary $2\DPDA$ \cite{Cook71},
was shown to be equivalent to $\log n$ space-bounded
auxiliary $2\NPDA$ which, in turn,
is equivalent to polynomial time-bounded $\DTM$.
Summarizing, we have:
\begin{proposition} \label{equiv1}
The languages accepted by the following coincide:
\begin{itemize}
\item multi-stack ordered $2\DCSPA$s,
\item multi-head multi-stack ordered $2\DCSPA$s,
\item multi-head $2\DPDA$s, 
\item $\log n$ space-bounded auxiliary $2\DPDA$s,
\item $\log n$ space-bounded auxiliary $2\NPDA$s,
\item polynomial time-bounded $\DTM$s.
\end{itemize}
\end{proposition}
Thus the class of languages accepted by
multi-stack ordered $2\DCSPA$s is exactly the class $\P$.
It is interesting to note that this model has no
restriction on the space usage in the pushdown
stack and the checking stacks.

The model of the multi-head multi-stack ordered $2\DCSPA$ has the 
restriction that the pushdown can only be used after
all the checking stacks are done writing.  
If this restriction is not present, this model accepts exactly the class of
elementary languages (those accepted by deterministic Turing machines with time complexity obtained by
iterating an exponential function) \cite{VogelWagner}.

\subsection{The Nondeterministic Case}

Next, we consider the nondeterministic version.
It is known that $2\NCSA$s (i.e.,  $1$-stack $2\NCSA$)
are equivalent to $n$ space-bounded $\NTM$s \cite{Ibarra-JCSS71}.
But, as the following proposition shows, $2$-stack $\NCSA$s
(hence, also $2$-stack $2\NCSA$s) are significantly more powerful.
\begin{proposition} \label{re}
A language $L$ is accepted by a $2$-stack $\NCSA$ if and only if $L$ is
accepted by a $\DTM$.
\end{proposition}
\begin{proof}
Let $M$ be a deterministic 2-counter machine with a one-way
read-only input.  It is well-known that such a machine is
equivalent to a $\DTM$ \cite{Minsky}.  Construct a $2$-stack $\NCSA$ $M'$
which, when given input $w$, first writes (without moving the input head) 
$1^n$ for nondeterministically
chosen $n$ ($n$ is guessed to be bigger than or equal to the maximum value placed on any counter) on the two stacks.  $M'$ then simulates the computation
of $M$ on input $w$ using the two stacks to simulate the counters.
That a 2-stack $\NCSA$ can be simulated by an $\NTM$ and, hence by
a $\DTM$, is obvious. 
\end{proof}

Hence, two-way multi-head multi-stack $2\DCSA$ is significantly less powerful
(all languages accepted are in $\P$)
than even one-way one-head $2$-stack $\NCSA$ (which accept all recursively enumerable 
languages), demonstrating the power of nondeterminism with
this model.

\subsection{A Hybrid Case}

Now we consider a restricted version of multi-head multi-stack $2\NCSA$.  
A {\it write-deterministic}
$r$-head $k$-stack $2\NCSA$ is one where the machine operates deterministically until every stack has
entered the reading phase, and thereafter can work nondeterministically.   Then we have:
\begin{proposition} \label{prop-write}$~~$
\begin{enumerate}
\item Let $h \ge 2$.  For every $h$-head $2\NFA$ $M$ and $r \ge 1, k \geq 1$ such that $h = r+k$, there is a
write-deterministic $r$-head $k$-stack  $2\NCSA$ $M'$ such that $L(M') = L(M)$.
\item Let $r \ge 1, k\ge 1$.  For every write-deterministic  $r$-head $k$-stack $2\NCSA$ $M$, there is
a $(3^kr+1)$-head $2\NFA$ $M'$ such that $L(M') = L(M)$.
\end{enumerate}
It follows that write-deterministic multi-head multi-stack $2\NCSA$
languages are in $\NLOG$ and, hence, in $\P$.
\end{proposition}
\begin{proof}  For part 1, given $M$, we construct $M'$ which when given input $w$, first copies $w$ on each of
the $k$ stacks.  Then $M'$ simulates $M$ using the $r$ input heads (on $w$) and the $k$ stacks (each containing $w$).
For part 2,  we can again use $r$-head $k$-stack $2\NCSA$s with sensing. By applying Lemma \ref{normalform}, the resulting machine is a write-deterministic $r$-head $k$-stack $2\NCSA$
in normal form. Then one can
verify that the proofs of Lemma \ref{lem-red} through Corollary \ref{alltoDFA} apply directly.
\end{proof}

\begin{remark}
Denote the nondeterministic version of an $r$-head $k$-stack $2\DCSPA$
by $r$-head $k$-stack $2\NCSPA$. We can then restrict this
nondeterministic model to be {\it write-deterministic}.  Then
an analogue of Proposition \ref{prop-write} holds when
$2\NFA$ and $2\NCSA$ are replaced by $2\NPDA$ and 
ordered $2\NCSPA$, respectively. It follows that write-deterministic multi-stack 
ordered $2\NCSPA$ is equivalent to multi-head $2\NPDA$ which is equivalent to 
$\log n$ space-bounded auxiliary $2\NPDA$ and which, in turn, is
equivalent to $\log n$ space-bounded $2\DPDA$
(by Proposition \ref{equiv1}). Thus, the models, 
write-deterministic multi-stack ordered $2\NCSPA$ and multi-head $2\NPDA$, can be 
added to the list of equivalences in Proposition \ref{equiv1}.
\end{remark}

We can define a {\it read-deterministic} $r$-head $k$-stack $2\NCSA$
as one where the machine can
operate nondeterministically until every stack
has entered the reading phase, and thereafter
must work deterministically.   However, it is easy to show that a read-deterministic $1$-head $1$-stack $2\NCSA$ $M$ can simulate an $n$ space-bounded $\NTM$ $Z$: $M$ simply guesses and writes on its stack a sequence of IDs on $Z$ on an input of length $n$ and deterministically checks that the sequence is an accepting computing. This is done by using the two-way input head to compare each position of each ID with corresponding positions of the next ID.  Similarly, any recursively enumerable language
can be accepted by a read-deterministic $1$-head $2$-stack $2\NCSA$ (in fact, even with one-way
input) as seen in the proof of Proposition \ref{re}.

\section{Synchronous Multi-Stack $2\DCSA$ and $2\NCSA$}  

Here, we will look at $k$-stack $2\DCSA$ ($2\NCSA$)
where the stacks do not have separate reading phases.
Thus, on any computation on any
input (accepted or not), when a stack enters the
reading phase, all other
stacks can no longer write.  We call such a machine {\em synchronous} and we refer
to these machines by
{\em $k$-stack synchronous $2\DCSA$ ($2\NCSA$)}. 
Note that $1$-stack synchronous $2\DCSA$ ($1$-stack synchronous $2\NCSA$)
is the same as $2\DCSA$ ($2\NCSA$).

\subsection{The Deterministic Case}

Clearly, languages accepted by multi-stack synchronous $2\DCSA$s
can be accepted by multi-stack $2\DCSA$s.
We will show perhaps surprisingly below that the converse is also true.

Proposition \ref{Engel} (which is a result
from \cite{EngelfrietCheckingStack}) states:
a language $L$ can be accepted by a $2\DCSA$
if and only if it can be accepted 
by a 2-head $2\DFA$.
We can generalize this result as follows:

\begin{proposition} \label{Prop32}
Let $k \geq 1$. A language $L$ is accepted by a $k$-stack synchronous $2\DCSA$ 
if and only if $L$ can be accepted by a $(k+1)$-head $2\DFA$.
\end{proposition}
\begin{proof}
Let $M$ be
a $(k+1)$-head $2\DFA$ $M$. We construct a
$k$-stack synchronous $2\DCSA$ $M'$ to simulate $M$ by first
copying the input $x$ into the $k$ stacks and then simulating $M$
using one head on the input and $k$ heads on the $k$ stacks.

To  prove the converse, we generalize
the construction of Proposition \ref{Engel} Part 1 in \cite{EngelfrietCheckingStack}. 
Let $M$ be a $k$-stack synchronous $2\DCSA$.
Let $\Sigma$  be the input alphabet and $\Gamma_1, \ldots, \Gamma_k$
be the disjoint stack alphabets of stacks $S_1, \ldots, S_k$, respectively.
Let $\Gamma = \Gamma_1 \cup \cdots \cup \Gamma_k$.

First construct a $2\GSM$ 
$T$ which for every input $w \in \Sigma^*$ 
of $M$,  outputs $T(w) = w \# z$,  where $z \in \Gamma^*$ is
the string which represents 
the ``combined'' strings $M$ writes on its $k$ stacks.  $T$ does this
by first outputting the input $w$ followed by $\#$ and then simulating
the writing phase of $M$.  So, e.g., if in a step, $M$ writes $s_i$  on
stack $S_i$  (where $s_i$ is a symbol or $\lambda$), then in the simulated
step, $T$ outputs $s_1 \cdots s_k$.  When $M$ completes the writing phase,
$T$ enters an accepting state.

Now construct a $(k+1)$-head $2\DFA$ $M'$ which when given input $w \# z$
first simulates the writing phase of $M$ on input $w$ by using one 
head on $w$ and another head to check that $z$ is a string representing
the ``combined''  strings written on the $k$ stacks. Then $M'$ simulates 
the reading phase using one head on $w$ and $k$ heads on $z$.  Note 
that when working on $z$, the input head simulating stack $S_i$ skips all 
symbols in $(\Gamma - \Gamma_i)$.  Then $T^{-1} (L(M')) = L(M)$. The result 
follows since languages accepted by $k$-head $2\DFA$s are closed under
inverse $2\GSM$s \cite{EngelfrietCheckingStack}.
\end{proof}


\begin{proposition} \label{hier}
$~~$
\begin{enumerate}
\item $(k+1)$-stack synchronous $2\DCSA$s are more powerful than
$k$-stack synchronous $2\DCSA$s.
\item The languages accepted by the following
models coincide:
\begin{multicols}{2}
\begin{itemize}
\item multi-stack synchronous $2\DCSA$s, 
\item multi-stack $2\DCSA$s,
\item multi-head multi-stack $2\DCSA$s,
\item multi-head $2\DFA$s, 
\item $\log n$ space-bounded $\DTM$s.
\end{itemize}
\end{multicols}
\end{enumerate}
\end{proposition}
\begin{proof} 
The first part follows from Proposition \ref{Prop32} and
the fact that $(k+1)$-head $2\DFA$s
are more powerful than $k$-head $2\DFA$s
(even on over unary alphabet) \cite{Monien}.
Part 2 follows
from Proposition \ref{Prop32} and Corollary \ref{equiv}.
\end{proof}

Specifically, to convert a $k$-stack $2\DCSA$ to a synchronous machine, the following increase in stacks
is sufficient:
\begin{corollary} \label{sim1}
Let $k \ge 2$.  If $L$ is accepted by a $k$-stack $2\DCSA$, then $L$ can be
accepted by a $(2 \cdot 3^{k-1})$-stack synchronous $2\DCSA$.
\end{corollary}
\begin{proof}
If $L$ is accepted by a $k$-stack $2\DCSA$,  then $L$ can be
accepted by a $(2 \cdot 3^{k-1} +1)$-head $2\DFA$
(by Proposition \ref{prop-red2}, part 1).
Then $L$ can be accepted by a $(2 \cdot 3^{k-1})$-stack synchronous $2\DCSA$
(by Proposition \ref{Prop32}).
\end{proof}

It is an interesting open question whether the number of additional stacks
needed in the simulation in Corollary \ref{sim1}  can be reduced and,
in particular, whether $6$-stacks can be reduced when $k=2$.
However, we conjecture that for $k \ge 2$,  there are languages accepted
by $k$-stack $2\DCSA$s that cannot be accepted by $k$-stack synchronous $2\DCSA$s.

In \cite{EngelfrietCheckingStack}, the following was shown:
\begin{proposition} \cite{EngelfrietCheckingStack} For $k \ge 1$, a $(k+1)$-head $2\DFA$ can be simulated by a $k$-head $2\DCSA$, which in turn can be simulated by a
$(3k+1)$-head $2\DFA$.  
\end{proposition}

It was left open whether less than $(3k+1)$ heads can be used to
simulate a $k$-head $2\DCSA$ for $k \ge 2$. However, for $k =1$, 
a $2$-head $2\DFA$ can simulate a $1$-head $2\DCSA$ (Proposition \ref{Engel} Part 1).
But for multi-stack synchronous $2\DCSA$ and multi-head $2\DFA$, we have 
a tight result in Proposition \ref{Prop32}: $k$-stack synchronous $2\DCSA$s are equivalent
to $(k+1)$-head $2\DFA$s for $k \ge 1$.

Interestingly, there is a restricted version of a $k$-head
$2\DCSA$ that is equivalent to a $k$-stack synchronous $2\DCSA$
(which, by Proposition \ref{Prop32}, would also then be
equivalent to a $(k+1)$-head $2\DFA$). The  result 
is interesting in that it gives the exact
trade-off between a synchronous $2\DCSA$ with 1 input head and $k$ stacks
and a restricted $2\DCSA$ with $k$ input heads and 1 stack.

Define an {\em $r$-head $l$-write-restricted $2\DCSA$}, or simply
$r$-head write-restricted $2\DCSA$, to be an $r$-head $2\DCSA$, where
the writing phase consists of $1 \leq l \le r$ writing subphases where subphase $i$ only involves moving one input head $h_i$
when writing on the stack, i.e., the symbols written during
subphase $i$ depend only on the states and the symbols scanned by head $h_i$ and the other
heads do not move from their current position during the subphase (which for all heads that are not $h_1, \ldots, h_i$ is the left end-marker).  At the end of
subphase $i$, head $h_i$  remains at its current position until the end of the write phase. Then the writing continues with
subphase $i+1$ with a different head.   After writing subphase $l$, the machine enters the
reading phase with all heads participating in the computation.
Clearly, for $r = 1$, the restricted version is the same as the unrestricted version.  
An $r$-head $1$-write-restricted $2\DCSA$ only allows one input head to be used during the write phase. The first lemma shows that any number of
phases can be reduced to one.

\begin{lemma} \label{new1}
A language $L$ is accepted by an $r$-head write-restricted $2\DCSA$ if and only if it can
be accepted by an $r$-head $1$-write-restricted $2\DCSA$.
\end{lemma}
\begin{proof}
We need only to show that an $r$-head $l$-write-restricted $2\DCSA$ $M$ where $r \ge 2$ and
$l \ge 2$ can be simulated by
an $r$-head $1$-write-restricted $2\DCSA$ $M_1$.  Let the input heads of $M$ be
$h_1, \ldots, h_r$ and $\#_1, \ldots, \#_l$ be new symbols.
We shall refer to the writing subphases of $M$ as 
$h_1$-subphase, $\ldots$, $h_l$ subphase. 
The corresponding input heads of $M_1$ are also called $h_1, \ldots, h_r$.
The single writing subphase of $M_1$ will only involve head $h_1$.

$M_1$ begins by simulating the $h_1$-subphase of $M$.  When the subphase
ends, the stack will contain some string $w_1$ and the input head
$h_1$ will be in some position on the input, which is some
distance (number of cells) $d_1$ from the left end-marker.
$M_1$ then moves $h_1$ to the left end-marker while
writing $\#_1^{d_1}$ on the stack.
The stack of $M_1$ will then contain $w_1 \#_1^{d_1}$.

Then $M_1$, again using
$h_1$, simulates the $h_2$-subphase of $M$ just like above, after which,
the stack of $M_1$ will contain $w_1 \#_1^{d_1} w_2 \#_2^{d_2}$.   
The process is repeated until the $h_l$-subphase of $M$ has been simulated.
The stack of $M_1$ will then contain $w_1 \#_1^{d_1} w_2 \#_2^{d_2}
\cdots w_l \#_l^{d_l}$. 

Next, $M_1$ enters the stack and, for each $1 \le i \le l$,
uses the segment $\#_i^{d_i}$ of the 
stack to restore input head $h_i$ to the position it was at on the 
input when the $h_l$-subphase ended. (Note that heads
$h_{l+1}, \ldots, h_{r}$, which were not involved
in the writing phase, are on the left end-marker.) 

Finally, $M_1$ simulates the reading phase of $M$. However, in the
simulation, $M_1$ ignores the symbols $\#_1, \ldots, \#_l$ 
on the stack. 
Clearly, $L(M) = L(M_1)$.
\end{proof}

\begin{proposition} \label{new2}
Let $k \geq 1$. The languages accepted by the following models coincide:
\begin{itemize}
\item
$(k+1)$-head $2\DFA$
\item
$k$-stack synchronous $2\DCSA$
\item
$k$-head write-restricted $2\DCSA$
\end{itemize}
\end{proposition}
\begin{proof}
The equivalence of the first two models above has already been established
in Proposition \ref{Prop32}.  That the first model can be simulated by the 
third model is obvious (just copy the input onto the stack and 
use the stack to simulate one head). There remains to show that
any language accepted by a $k$-head write-restricted $2\DCSA$ $M$ can be simulated by
a $k$-stack synchronous $2\DCSA$ $M'$. By Lemma \ref{new1}, we can assume that
$M$ is a $k$-head $1$-write-restricted $2\DCSA$. Let
the input heads of $M$ be $h_1, \ldots, h_k$, and $h_1$ is
the head involved in the single writing phase. $M'$ has
a single head $h$ and checking
stacks $s_1, \ldots, s_k$, and operates as follows:
$M'$ first copies the input into stacks $s_2, \ldots, s_k$.
Then it simulates the single writing phase of $M$ on
stack $s_1$.  Finally,
$M'$ simulates the reading phase of $M$ using stack $s_1$ and the
input and stacks $s_2, \ldots, s_k$ to mimic the $k$ input heads of $M$.
\end{proof}

From Propositions \ref{hier} and \ref{new2}, we have: 
\begin{corollary} \label{new3}
$(r+1)$-head write-restricted $2\DCSA$s are more powerful than $r$-head write-restricted $2\DCSA$s.
\end{corollary}

\begin{remark}
As above, we can define a synchronous version of the
$k$-stack ordered $2\DCSPA$. 
We can also define an {\em $l$-write-restricted ordered $2\DCSPA$}.
Again all the results above hold when
$2\DFA$, $2\DCSA$, synchronous $2\DCSA$, $\log n$ space-bounded $\DTM$, and
$r$-head $l$-write-restricted $2\DCSA$ 
are replaced by 
$2\DPDA$, ordered $2\DCSPA$, synchronous ordered $2\DCSPA$, 
$\log n$ space-bounded auxiliary $2\DPDA$,
and $r$-head $l$-write-restricted ordered $2\DCSPA$, respectively, noting that 
$k$-head $2\DPDA$ is weaker than $(k+1)$-head
$2\DPDA$ \cite{Ibarramultihead}.     
\end{remark}

\subsection{The Nondeterministic Case}

Next, we consider $k$-stack synchronous $2\NCSA$.
The characterization is simple, as we have:
\begin{proposition} \label{2rncsa}
Let $k \ge 1$. A language $L$ is accepted by a $k$-stack synchronous $2\NCSA$ if and only if
$L$ is accepted by a $k$-stack $2\NCSA$.
\end{proposition}
\begin{proof}
It is sufficient to show that for $k \ge 2$,
every $k$-stack $2\NCSA$ $M$ can be simulated
by a $k$-stack synchronous $2\NCSA$ $M'$.

When given an input $w$, $M'$ first nondeterministically
writes strings $z_1, \ldots, z_k$ on the $k$ stacks 
while reading $\lambda$ on the input.  Then $M'$
moves all the stack heads to the bottom of their respective
stacks.  Then $M'$ simulates $M$ on $w$ while checking
that each $z_i$ is the string written by $M$ on stack $i$
during the writing phase.
\end{proof}

It follows from Proposition \ref{2rncsa} that all the results
concerning $2\NCSA$s in the previous section apply to synchronous $2\NCSA$s.

Finally, we can define an $r$-head write-restricted $2\NCSA$ as the nondeterministic 
version of an $r$-head write-restricted $2\DCSA$. Using the construction in the proof
above, it is clear that an $r$-head write-restricted $2\NCSA$ is equivalent to
an $r$-head $2\NCSA$ which, in turn, has been shown 
to be equivalent to an $n^r$ space-bounded $\NTM$
\cite{Ibarra-JCSS71}.

\subsection{Some Decidability Questions Concerning Synchronous Machines}
\label{sec:phasedmachines}

Here, we will address decidability questions regarding whether $k$-stack
machines are synchronous.
%

\begin{proposition} \label{separate-1}
For any $k \ge 2$, it is undecidable, given a $k$-stack $2\DCSA$ $M$,
whether it is synchronous. 
\end{proposition}
\begin{proof}
Clearly, we only need to prove the case when $k =2$.
It is known that it is undecidable, given a deterministic 2-counter machine $Z$ with
no input starting from both counters being zero,
whether it will halt \cite{Minsky}. Let $Z$ be such a 2-counter machine.
We construct a
2-stack $2\DCSA$ $M$ over unary input that is synchronous if and only if $Z$ does not halt. $M$ has
stacks $S_1$ and $S_2$ which, when given
a unary input string $1^n$, first copies $1^n$ in $S_1$ and
writes $\#$ in $S_2$.  Then $M$ simulates $Z$ using the
input head and $S_1$ (just by moving both tape heads) but can only continue the simulation if neither
counter exceeds $n$. Hence, as $M$ is simulating $Z$, $M$
has already started reading from $S_1$; and if $M$ will later write to $S_2$, then
$M$ is not synchronous.
If $Z$ halts without any of its counters
exceeding the value $1^n$, $M$ writes another $\#$ in $S_2$
and accepts.  If $Z$ exceeds the value $n$ in one of
its counters, $M$ detects this and accepts without modifying $S_2$.
Clearly, if $Z$ goes into an infinite loop without exceeding
the value $n$ in any of its counters, $M$ also goes into
an infinite loop. It follows that $M$ is synchronous if
and only if, for every input $1^n$ to $M$, $M$ only writes one $\#$ if and only if
there is no $n$ where $Z$ halts with both counters at most $n$, if and only if 
$Z$ does not halt, which is undecidable.
\end{proof}

Proposition \ref{separate-1} still holds 
when the machine has a one-way input for any $k \ge 3$: 
\begin{corollary} \label{separate-2}
For any $k \ge 3$, it is undecidable, given a $k$-stack $\DCSA$ $M$,
whether it is synchronous. 
\end{corollary}
\begin{proof}
$M$ with stacks $S_0, S_1, S_2$ first
copies its one-way unary input $1^n\#$ (where $\#$ is the
right end-marker for the one-way input) in $S_0$.
Then it proceeds as in the proof of Proposition \ref{separate-1}
with $S_0$ acting as the two-way input.
\end{proof}

Note that Proposition \ref{separate-1} and Corollary \ref{separate-2}
also hold for nondeterministic machines.
We now show that Corollary \ref{separate-2} does not hold
for $k =2$, even when the machine is nondeterministic.
\begin{proposition} \label{separate-3}
It is decidable, given a 2-stack $\NCSA$ $M$,
whether it is synchronous.
\end{proposition}
\begin{proof}
Given a 2-stack $\NCSA$ $M$ with stacks
$S_1$ and $S_2$, we will construct a $2\NFA$
$M'$ such that $M$ is synchronous 
if and only if $L(M') = \emptyset$, which
is decidable.

Let $\Sigma$ be the input alphabet of $M$,
and $\Gamma$ be the stack alphabet of $S_1$ and $S_2$.
The input to $M'$ is a string 
$w \in \Gamma^*$ (with left and right end-markers).
$M'$ when given $w$ first guesses
which of the two stacks enters the reading
phase first.

Suppose $M'$ guesses that $S_1$ enters
the reading phase first.  $M'$ guesses
the one-way input $x\lhd$ ($\lhd$ the end-marker) to $M$  (where $x  \in \Sigma^*$)
symbol-by-symbol
and simulates the writing and reading phases
of stack $S_1$ (on input $x\lhd$) by verifying that it writes $w$
(the input to $M'$) during the write stage while not applying any transitions that read from $S_2$,
then simulating the read stage of $M$ by using the two-way input $w$,
 until either:

{\em Case 1}: $S_2$ enters
its stack after $S_1$ has entered its stack, or

{\em Case 2}: $S_2$ writes a symbol after $S_1$ has
entered its stack.

\noindent
Note that $M'$ can guess
the end-marker $\lhd$ (which is the last
symbol of $x\lhd$) either during the 
writing  phase of  $S_1$ or during the
reading phase of $S_1$. 
If Case 1 occurs, then $M'$ rejects $w$.
If Case 2 occurs, then $M'$ accepts $w$.

The case when $M'$ guesses that $S_2$ enters
the reading phase first is handled similarly.
Clearly, $M$ is synchronous
if and only if $L(M') = \emptyset$.
\end{proof}

\section{Space Complexity of Multi-Stack $2\NCSA$}

This section will consider two measures of space complexity for
multi-stack $2\NCSA$s.
\begin{definition}
Let $k \geq 1$, and let $M$ be a $k$-stack $2\NCSA$. For $S(n)$ a non-zero function, define the following space complexity measures.
\begin{itemize}
\item We say $M$ is $S(n)$-limited if, for every $w \in L(M)$ with $w$ of length $n$, there is an accepting computation in which the string written on stacks $2$ through $k$ have length at most $O(S(n))$.

\item We say $M$ is $S(n)$-bounded if, for every $w \in L(M)$ with $w$ of length $n$, there is an accepting computation in which the string written on stacks $1$ through $k$ have length at most $O(S(n))$.
\end{itemize}
\end{definition}
Notice that with the definition of $S(n)$-limited, there is one stack that is unbounded, while the others are all at most $S(n)$ in size; whereas with $S(n)$-bounded, every stack is at most $S(n)$ in size.

\subsection{Characterizations of $S(n)$-Limited Multi-Stack $2\NCSA$s}

We will show that a $k$-stack $S(n)$-limited $2\NCSA$s can be characterized in
terms of space-bounded $\NTM$s.
We will use the following characterization 
of $r$-head $2\NCSA$ (i.e., a 1-stack $2\NCSA$ with $r$ two-way heads)
in \cite{Ibarra-JCSS71}:

\begin{proposition} \cite{Ibarra-JCSS71} \label{JCSS71}
\label{NCSANTM}
Let $r \ge 1$.  A language $L$ can be accepted by an $r$-head $2\NCSA$
if and only if $L$ can be accepted by an $n^r$ space-bounded $\NTM$.
\end{proposition}


\begin{proposition} \label{nlogn} Let $S(n)$ be a non-zero function.
A language $L$ is accepted by a $2$-stack $S(n)$-limited $2\NCSA$
if and only if $L$ can be accepted by an $nS(n)$ space-bounded $\NTM$.
\end{proposition}
\begin{proof}
Let $Z$ be an $nS(n)$ space-bounded $\NTM$. We can assume without loss of
generality that $Z$ is an $\NTM$ with a single read/write worktape that initially contains the input of size $n$, and can only grow to the right, since $Z$ allows at least linear space.  Construct
a 2-stack $S(n)$-limited $2\NCSA$ $M$ with stacks $S_1$ (the unbounded stack) 
and $S_2$ (the $S(n)$-bounded stack) which,
when given an input $w$ of length $n$, operates as follows:
\begin{enumerate}
\item
$M$ writes $ID_1 \# \cdots \#ID_s$ on stack $S_1$ for some
nondeterministically
chosen configurations $ID_1, \ldots , ID_s$
of the worktape of $Z$ (a configuration is a string
$uqv$, where $uv$ is a string over the worktape alphabet of $Z$
and $q$ is a state and also the position of the read/write head), and $ID_1$ is
an initial ID starting with the input encoded, and $ID_s$
is an accepting $ID$. It is clear that we can assume each ID is of the same length $m \leq n S(n)$
by padding with blanks and later $M$ will verify they are of the same length.
\item
$M$ writes a unary string $1^t$ on $S_2$, where $t$ is nondeterministically
guessed. 
\item
$M$ checks that each $ID$ on $S_1$ is of length  $nt$.
This can be 
done using the input of length $n$ and stack $S_2$ of length $t$
to count up to $nt$. 
\item
Finally, $M$ checks that $ID_{i+1}$ is a valid successor of $ID_i$
of $Z$ for $1 \le i \le s-1$, again using the input
and stack $S_2$ to repeatedly count to $nt$ thereby enabling the ability to check
that positions of ID$_i$ match ID$_{i+1}$ except close to the read/write head where it verifies that it is changing via a transition of $Z$.
\end{enumerate}
Note that if $w$ (of length $n$) is accepted by $Z$, then it
has an accepting computation that uses at most $nS(n)$ space.
It follows that in step (2) above, there is an $t \le S(n)$ which
would allow $M$ to simulate $Z$ successfully.  Hence, $L(M) = L(Z)$.

To prove the converse, we will use Proposition \ref{JCSS71} (with $r=1$) and translation.
Let $M$ be a 2-stack $S(n)$-limited $2\NCSA$ $M$
with input alphabet $\Sigma$ and stack alphabets $\Gamma_1$ and $\Gamma_2$.
Assume that these alphabets are disjoint, and 
$S_1$ is the unbounded stack, and $S_2$ is the bounded stack.
Let $\#$ be a new symbol.

We first construct a $2\NCSA$ $M'$ (this has only
stack $S_1$), which operates as follows when given input $w$:
\begin{enumerate} 
\item
$M'$ checks that the input is of the form
$w = a_1\#u_1\#a_2\#u_2 \cdots a_n\#u_n$, where
$a_i \in \Sigma$, $u_i \in \Gamma_2^*$.
\item
$M'$ writes a string $z = u_1 \#^{r_0} z_1 \#^{r_1}z_2\#^{r_2} \cdots z_k \#^{r_k}$
on $S_1$, where $z_i \in \Gamma_1, r_i$, and $k$ are nondeterministically chosen.
\item
$M'$ checks that $u_1 = \cdots = u_n$ and $r_1 = \cdots = r_k = |u_1|$.
\item
$M'$ simulates $M$ on input $a_1 \cdots a_n$
and checks that $u_1$ was the string written on the bounded stack
$S_2$ of $M$ and $z_1 \cdots z_k$ was the string written on the 
unbounded stack $S_1$ of $M$.  This is done as follows:

Suppose $M$ is on symbol $a_i$ on the input and has so far written
a string of length $p$ on stack $S_2$ ($p \ge 0)$
and written $z_1 \cdots z_j$ on stack $S_1$ ($j \ge 0$).
 This is represented in $M'$ with its
input head in position $p$ of $u_i$ which is directly
to the right of $a_i$ (if $p = 0$, the head is on $\#$)
and the stack head of $S_1$ on $z_j$ (if $j = 0$, the head is on $\#$). 

We now describe the simulation of a move of $M$.
If, e.g.,  $M$ moves its input head to the left and writes a $(p+1)^{\mbox{st}}$ 
symbol in $S_2$ and writes $z_{j+1}$ in $S_1$, then $M'$ performs
the following sequence of moves:
\begin{enumerate}
\item
$M'$ uses the segment $\#^{r_j}$ in $S_1$ (which is
at the right of $z_j$) to temporarily
remember the number $p$ by moving the input head to the left. 
\item
$M'$ moves its input head to $a_{i-1}$ and
uses the number $p$ remembered in $S_1$ to move the input
head to the $p+1^{\mbox{st}}$ symbol of $u_{i-1}$ and checks
that the symbol in that position is correct. 
\item
$M'$ moves its stack head 
right and checks that $z_{j+1}$ is the correct symbol.
\end{enumerate}

The other moves of $M$ are treated similarly.
\end{enumerate}
$M'$ accepts if and only $M$ accepts.  
Since $M$ has only one stack, $L(M')$
can be accepted by an $|w|$ space-bounded $\NTM$ $Z'$
by Proposition \ref{NCSANTM}.

Finally, we construct from $Z'$ another $\NTM$ $Z$ which,
on input $x = a_1 \cdots a_n$,
writes on a read/write tape $w =  a_1\# u_1\# a_2\# u_2 \cdots  a_n \#u_n$,
where $u_1, \ldots, u_n$ are nondeterministically
guessed strings in $\Gamma_2^*$.  $Z$ then simulates $Z'$ on $w$.
It follows that $Z$ accepts $L(M)$. Also notice that for every
input $x= a_1 \cdots a_n$ to $M$, there is some accepting computation on $x$ where the
second stack is bounded by $S(n)$. Hence, on input $a_1 \cdots a_n \in L(M)$, $Z$ can accept with a computation
that writes a word $w$ on the tape that is of length at most $2n + |u_1|n$, where 
$|u_1| = \cdots = |u_n| \leq S(n)$, and so $|w| \leq 2n + nS(n)$. Hence, $Z$ is $nS(n)$ space-bounded.
\end{proof}

Generalizing Proposition \ref{nlogn}, we have:
\begin{proposition} \label{nlogn1}
Let $k \ge 1$ and let $S(n)$ be a non-zero function. A language $L$ is accepted by a
$(k+1)$-stack $S(n)$-limited $2\NCSA$ if
and only if $L$ can be accepted by an $nS^k(n)$ space-bounded
$\NTM$.
\end{proposition}
\begin{proof}
We need only to prove the result when $k \ge 2$.
As in the first part of the proof of
Proposition \ref{nlogn},
an $nS^k(n)$ space-bounded $\NTM$ $Z$ can be simulated 
by a $(k+1)$-stack $S(n)$-limited $2\NCSA$ $M$ with
stacks $S_1, S_2, \ldots, S_{k+1}$, where each of
$S_2, \ldots, S_k$ will now have unary strings of length $t \le S(n)$,
where $t$ is nondeterministically guessed.
These stacks, together with the input, can count up
to $nt^k$ and could check whether a sequence of $ID$'s
written on $S_1$ (each $ID$ is of length $nt^k$)
is a valid accepting computation. 

For the converse, we use induction.  Assume that
for $k \ge 1$, a $(k+1)$-stack $S(n)$-limited $2\NCSA$
can be simulated by an $nS^k(n)$ space-bounded $\NTM$.
Let $M$ be a $(k+2)$-stack $S(n)$-limited $2\NCSA$
with stacks $S_1, S_2, \ldots, S_{k+2}$, where $S_1$
is unbounded and the other stacks are bounded.  
We construct a $(k+1)$-stack $S(n)$-limited
$2\NCSA$ $M'$ which ``encodes'' $S_2$ in the input
as in the proof of Proposition \ref{nlogn}. Then, by
induction hypothesis, we construct from
$M'$ an $nS^k(n)$ space-bounded
$\NTM$ $Z'$ such that $L(Z') =L(M')$. 
Finally, we construct from $Z'$ an $nS^{k+1}(n)$
space-bounded $\NTM$ $Z$ accepting $L(M)$, as in 
Proposition \ref{nlogn}.
\end{proof}

We can use space hierarchy results for $\NTM$s to show hierarchies
of multi-stack space-limited $2\NCSA$s.
For example, since more languages can be accepted with $n\log^{k+1} n$
space-bounded $\NTM$s than $n\log^k n$ space-bounded
$\NTM$s \cite{Seiferas1977}, we have:
\begin{corollary}
Let $k \geq 1$. There are languages accepted by $(k+1)$-stack $\log n$-limited
$2\NCSA$s that cannot be accepted by $k$-stack $\log n$-limited
$2\NCSA$s.
\end{corollary}

Similarly, since $n^{k+1}$ space-bounded $\NTM$s are computationally more
powerful than $n^k$ space-bounded $\NTM$s \cite{Ibarra-JACM1972}
for any $k \ge 1$, we have:

\begin{corollary}
Let $k \geq 1$. There are languages accepted by $(k+1)$-stack $n$-limited
$2\NCSA$s that cannot be accepted by $k$-stack $n$-limited
$2\NCSA$s.
\end{corollary}

From Propositions \ref{nlogn} and \ref{nlogn1}, we have:
\begin{corollary}
\label{cor2}
The following are equivalent for a language $L$:
\begin{itemize}
\item
$L$ is accepted by a multi-stack polynomial-space-limited $2\NCSA$.
\item
$L$ is accepted by a multi-stack $n$-space-limited $2\NCSA$.
\item
$L$ is accepted by a 2-stack polynomial-space-limited $2\NCSA$.
\item
$L$ is in $\PSPACE$.
\end{itemize}
\end{corollary}

In the future, we plan to investigate whether the power of $S(n)$-limited machines changes if all branches respect the
space bound rather than some accepting branch of each word.

\subsection{Characterizations of $S(n)$-Bounded Multi-Stack $2\NCSA$s}
\label{sec:alllinear}

In this subsection, we will consider $S(n)$-bounded $2\NCSA$s. For the important special
case of $n$-bounded machines, we also refer to these as linear-bounded.
\begin{proposition} 
\label{linear1}
Let $k \geq 1$. For every $(k+1)$-head $2\NFA$ $M$,
there is a $k$-stack linear-bounded $2\NCSA$ $M'$ such that $L(M') = L(M)$.
\end{proposition}
\begin{proof}
$M'$ copies the input $x$ into the $k$ stacks and then simulates $M$
using one head on the input and $k$ heads on the $k$ stacks.
\end{proof}
This result holds for the deterministic case as well.

Unfortunately, the converse of Proposition \ref{linear1}
is unlikely, since multi-head $2\NFA$ languages are
in $\NLOG$ (which is a subset of $\P$)
but linear-bounded $\NCSA$s can accept $\NP$-complete
languages \cite{ShamirBeeri} (see also Proposition \ref{detSAT}).  However, we can prove a characterization 
involving homomorphisms.

A {\em letter-to-letter} homomorphism $h: \Sigma^* \rightarrow \Gamma^*$
is a homomorphism where for all $a \in \Sigma$, $h(a)$ is in $\Gamma$.
A $\lambda$-free homomorphism $h$ is one where $h(a) \neq \lambda$ for all $a \in \Sigma$.

\begin{proposition} \label{linear2}
$~~$
\begin{enumerate}
\item
Let $k \ge 2$. For every $k$-head $2\NFA$ $M$
and letter-to-letter homomorphism $h$,
there is a $k$-stack linear-bounded $\NCSA$ 
 $M'$ such that $L(M') = h(L(M))$.
(Note that if $k = 1$, $M$ accepts a regular set
and regular sets are closed under homomorphism, so
$M'$ does not need a stack.)
\item
Let $k \ge 1$. For every $k$-stack linear-bounded $2\NCSA$ $M$,
there is a $(k+1)$-head $2\NFA$ $M'$ and a letter-to-letter
homomorphism $h$ such that $L(M) = h(L(M'))$.
\end{enumerate}
\end{proposition}
\begin{proof}
For the first part, $M'$ when given input $x$,
guesses an input $y$ to $M$ (symbol-by-symbol)
and writes $y$ into the $k$ stacks while checking that 
$x = h(y)$.  Then $M'$ simulates $M$ on $y$ using
the $k$ stacks and  accepts $x$ if $M$ accepts $y$. 
Note that $M'$ is a $k$-stack linear-bounded $\NCSA$
(i.e., one-way).

%
%

For the second part, let $c$ be a positive integer
such that the heights of the $k$-stacks during the writing phase
is at most $cn$, where $n$ is the input length.
Let $\Sigma$ be the input alphabet of $M$.  Assume that
the stacks have a common stack alphabet $\Gamma$.
The input to $M'$ will have
$(k+1)$ tracks:  Track $0$ contains the input string
 $w$ to $M$ and for $1 \le i \le k$,
track $i$ contains the
string $x_i$ written on stack $i$ during the
writing phase of stack $i$ (note that the
stacks have separate reading phases).

Clearly, we can partition each string $x_i$
on track $1 \le i \le k$ as a concatenation
of substrings of length at most $c$.  Thus,
$x_ i = y_{i_1} \cdots y_{i_{m_i}}$, where
we can assume $m_i = n$ as the substrings can be of length zero, and
also we can assume that the first subword $y_{i_j}$
that is smaller than length $c$ must have $y_{i_l} = \lambda$ for all $l>j$.

We can then define a $(k+1)$-track of symbols
of the form $[a, y_1, \ldots, y_k]$ where $a$
is in $\Sigma$ and each $y_i$ is
a string in $\Gamma^*$ of length at most $c$.
Denote the set of all such symbols by $\Gamma$.

Now construct a $(k+1)$-head $2\NFA$ which,
when given a string $z$ in $\Gamma^*$,
first checks that the strings on tracks
$1$ to $k$ are contiguous (if a cell contains a string of length less than $c$, all further cells
only contain $\lambda$'s).
Then $M'$ simulates the $2\NCSA$ $M$ using
one head to simulate the input head of
$M$ and the remaining $k$ heads to simulate
the $k$ checking stacks, making sure that
track $i$ contains exactly the string
written to stack $i$ during its writing 
phase.  $M'$ accepts if $M$ accepts.
Now let $h$ be a homomorphism that maps
every symbol $[a,y_1, \ldots, y_k]$ to $a$.
It is straightforward to verify that
$L(M) = h(L(M'))$.
\end{proof}

A letter-to-letter homomorphism is a special
case of $\lambda$-free homomorphism. Clearly, Proposition \ref{linear2}
still holds if we replace letter-to-letter homomorphism to
$\lambda$-free homomorphism.

Recall the definition of a one-way stack automaton from Section \ref{sec:prelims}, which is a generalization of
a one-way pushdown automaton where the stack head
can enter the stack in a two-way read-only mode.
When the stack head returns to the top of the stack
it can push and pop \cite{GGH}. The deterministic
(nondeterministic) version is denoted by $\DSA$ ($\NSA$).
Clearly, a $\DCSA$ ($\NCSA$) is a special case of a 
$\DSA$ ($\NSA$). In \cite{KingWrathall}, the following
was shown:
\begin{proposition} \label{stackautomata} \cite{KingWrathall} \label{KW}
If a language is accepted by a $\DSA$ $(\NSA)$, then
it is the image under a $\lambda$-free homomorphism of
a language accepted by a 
$\log n$ space-bounded $\DTM$ $(\NTM)$.
\end{proposition}

Since multi-head $2\DFA$s ($2\NFA$s) are equivalent to 
$\log n$ space-bounded $\DTM$s ($\NTM$s) \cite{Ibarra-JCSS71}, the following
results are, in some sense, tighter and more complete characterizations:
\begin{corollary} \label{chars}
$~~$
\begin{enumerate}
\item
A language $L$ is accepted by a multi-stack linear-bounded $2\DCSA$ if and only
if $L$ can be accepted by a $\log n$ space-bounded $\DTM$.
\item
A language $L$ is accepted by a multi-stack linear-bounded $2\NCSA$ if and only if
$L$ is the image under $\lambda$-free homomorphism of a language
accepted by a $\log n$ space-bounded 
$\NTM$.
\item
If language $L$ is accepted by an $\NSA$, then it can be accepted 
by a multi-stack linear-bounded $2\NCSA$ (hence, also by a multi-stack linear-bounded
$\NCSA$, since the former can be simulated by the latter
which first copies the input onto an additional stack and simulates
the former by using the additional stack as two-way input).
\end{enumerate}
\end{corollary}
\begin{proof}
Part 1 follows from Corollary \ref{equiv},
which clearly also holds for multi-stack linear-bounded $2\DCSA$s.
Part 2 follows from Proposition \ref{linear2} and since 
the languages accepted by multi-head $2\NFA$s
coincides with the languages
accepted by $\log n$ space-bounded
$\NTM$s \cite{Ibarra-JCSS71}.
Part 3 follows from Part 2 and Proposition \ref{stackautomata}.
\end{proof}

The next result gives an upper bound on the complexity of
multi-stack linear-bounded $2\NCSA$ languages.
\begin{proposition} \label{NP}
Every language accepted by a multi-stack linear-bounded $2\NCSA$
is in $\NP$.
\end{proposition}
\begin{proof}
Suppose $M$ is a $k$-stack linear-bounded $2\NCSA$ ($k \ge 1$).
By Proposition \ref{linear2}, part 2,
there is a $(k+1)$-head $2\NFA$ $M'$
and a $\lambda$-free homomorphism $h$ such that $L(M) = h(L(M'))$.
Now the family of languages accepted by multi-head $2\NFA$s
(which coincides with the family of languages
accepted by $\log n$ space-bounded
$\NTM$s \cite{Ibarra-JCSS71}) is contained in $\P$.
So to determine if a string $x$ of length $n$ is in
$L(M)$, we guess a string $y$ such that $x = h(y)$
and then determine if $y$ is in $L(M')$. It follows
that $L(M)$ is in $\NP$.
\end{proof}

In terms of one-way machines, the following was shown in \cite{ShamirBeeri}:
 \begin{proposition} \label{detSAT} \cite{ShamirBeeri}
There is a language $L$ accepted by a linear-bounded $\DCSA$ M
and a letter-to-letter homomorphism $h$ such
that $h(L)$ is $\NP$-complete. Further,
$h(L)$ is accepted by a linear-bounded $\NCSA$.
Hence, $\NCSA$s accept $\NP$-complete languages.
\end{proposition}

\begin{corollary}  
There is a language $L$ accepted by a 
$\DCSA$ and a letter-to-letter homomorphism $h$
such that if $h(L)$ is accepted by a $\log n$
space bounded $\NTM$, then
$\P = \NP$. Furthermore, there is a language $L'$ accepted by a 
$\DCSA$ and a letter-to-letter homomorphism $h$ such
that if $h(L')$ is accepted by a  multi-stack $2\DCSA$,
then $\P = \NP$.
\end{corollary}
\begin{proof}
The first statement follows from Proposition \ref{detSAT}
and the fact that the languages accepted by $\log n$
space-bounded $\NTM$s are in $\P$.
The second statement follows from the first statement and
Corollary \ref{equiv}.
\end{proof}
Even stronger, there is a language accepted by a 
$\DCSA$ and a letter-to-letter homomorphism $h$ such
that if $h(L)$ is accepted by a  multi-stack $2\DCSA$,
then $\DLOG = \NP$.

We have shown earlier that a language
$L$ is accepted by a multi-stack linear-bounded $2\DCSA$ if
and only if $L$ can be  accepted by a $\log n$
space-bounded $\DTM$. This result does not hold
for machines with one-way input:

\begin{lemma} \label{lem39} There is a language that 
can be accepted by a linear-bounded $\DCSA$ but
cannot be accepted by any one-way
$\log n$ space-bounded $\NTM$.
\end{lemma}
\begin{proof}
Let $L = \{x \# x^R ~|~ x \in (a+b)^*\}$. Then 
$L$ can be accepted by a linear-bounded
$\DCSA$ which, when given input
$x\#y$ (we assume the input has this form
since the finite-state control can check this),
copies $x$ on the stack and then checks
that $y = x^R$. But it is known (and easy to
prove) that $L$ cannot be 
be accepted by any one-way $\log n$ space-bounded
$\NTM$ $M$.  (On input $w = x\#x^R$, where $|x| = m$,
there are only polynomial in $m$ possible configurations
that $M$ can be in when its one-way input head reaches 
$\#$, but there are $O(2^m)$ strings of 
the form $x \# x^R$.)
\end{proof}

\begin{proposition}
There is a language that 
can be accepted by a linear-bounded $2\DCSA$ that only makes
one turn on its input that
cannot be accepted by any $\NCSA$ (linear-bounded or not).
\end{proposition}
\begin{proof}
Let $L = \{(a^n\#)^n ~|~ n \ge 1\}$. Assume that
$L$ is accepted by an $\NCSA$. The class of languages
accepted by $\NCSA$s is closed under homomorphisms \cite{CheckingStack}.
Thus, if we define $h(a) = a$ and $h(\#) = \lambda$, then 
$h(L) = \{ a^{n^2} ~|~ n \ge 1 \}$ can be accepted by an $\NCSA$, but this language
cannot be accepted by any machine in $\NCSA$ \cite{Greibach}.

However, a linear-bounded $2\DCSA$ $M$ can accept $L$:
$M$ copies
the first input segment $a^n\#$ into the stack
and using the stack and the two-way input head
checks that each segment is $a^n\#$, and the
number of input segments is $n$.
Note that $M$'s input head need only makes 
one turn on the input:  a left-to-right sweep
followed by a right-to-left sweep and accept if
the input is in $L$.
\end{proof}
The previous proposition can be strengthened from linear-bounded to $\log n$-bounded if we allow more turns
on the input tape.
\begin{proposition} \label{binarylanguage}
There is a language that can be accepted by a
$\log n$-bounded $2\DCSA$ but cannot be accepted by an $\NCSA$.
\end{proposition}
\begin{proof}
Let $L = \{ x_1 \# x_2 \# \cdots  \# x_r \#  ~|~  x_i$ is a binary number with
            the most significant bit on the right, $|x_1| = \cdots = |x_r| = m, 
            x_1 = 0^m, x_r = 1^m, x_{i+1} = x_i+1, m \ge 1 \}$.
So, e.g., $000\# 100\# 010\# 110\# 001\# 101\# 011\# 111\#$ is in $L$.

$L$ can be accepted by a $2\DCSA$ $M$
whose stack is $\log n$-bounded as follows:
$M$ writes a unary string of length $|x_1|$
in the stack  and then uses the stack to go back and forth between $x_i$ and $x_{i+1}$
to check that $x_{i+1} = x_i +1$.

If $L$ is accepted by an $\NCSA$, then since the class of languages
accepted by $\NCSA$s is closed under any homomorphism $h$ \cite{CheckingStack},
defining $h(0) = h(1) = \lambda$ and $h(\#) = a$,  
$h(L) = \{ a^{2^m} ~|~ m \ge 1 \}$ can be accepted by an $\NCSA$.
This contradicts a result in \cite{Greibach} where it was shown that 
any infinite unary language accepted by an $\NCSA$ must contain an 
infinite regular set.
\end{proof}

Now that we know that a $\log n$ $2\DCSA$ (hence, $2\NCSA$) can accept
a non-regular language, it is of interest to know whether there is a 
hierarchy in terms of the number of stacks.
For $k \ge 1$, let
$$L_k = \{ x_1 \# x_2 \# ...  \# x_r \#  ~|~ \begin{array}[t]{l}  x_i \mbox{~is a binary number with the most significant bit on}\\ 
\mbox{the right,~} |x_1| = \cdots = |x_r| = m^k, x_1 = 0^{m^k}, m^{th} \mbox{~bit of}\\ x_1 \mbox{~is marked},
x_r = 1^{m^k}, x_{i+1} = x_i+1, 1 \leq i < r, m \ge 1 \}.\end{array}$$
We can construct a $k$-stack $\log n$-bounded $2\DCSA$ $M_k$ that accepts $L_k$:
$M_k$ writes the first $m$ bits of $x_1$ into the $k$ stacks (this is possible since the $m^{\mbox{th}}$ bit is marked with a special character). Then $M_k$
uses the stacks to check that each $x_i$ has length $m^k$ and
that $x_{i+1} = x_i$.
We conjecture that the languages $L_1, L_2, \ldots$  form a hierarchy in terms of the number of
stacks, i.e.,  $L_k$ cannot be accepted
by any $(k-1)$-stack $\log n$-bounded $2\DCSA$ (or even $2\NCSA$) for $k \ge 2$.

\section{Conclusions}

We studied several generalizations of checking stack automata
and characterized their computing power in terms of
two-way multi-head finite automata and space-bounded Turing machines. 
Further, we proved hierarchies of these models with respect to their
recognition power in terms of the number of input heads and the number of checking stacks. 
Our characterizations and hierarchy results expanded/tightened
some previously known results.  We also investigated some space
complexity and decidability questions for the models introduced. Some
questions remain open, e.g., improving the stack-head trade-offs in the
conversions for some of the results.

\section*{Acknowledgments}
The research of O. H. Ibarra was supported, in part, by
NSF Grant CCF-1117708. The research of I. McQuillan was supported, in part, by Natural Sciences and Engineering Research Council of Canada Grant 2016-06172.

\bibliography{bounded}{}

\begin{thebibliography}{10}

\bibitem{AhoUllman1970}
A.~V. Aho, J.~D. Ullman and J.~E. Hopcroft, On the computational power of
  pushdown automata, {\em Journal of Computer and System Sciences} {\bf 4}(2)
  (1970)  129--136.

\bibitem{Cook71}
S.~A. Cook, Characterizations of pushdown machines in terms of time-bounded
  computers, {\em J. {ACM}} {\bf 18}(1)  (1971)  4--18.

\bibitem{EngelfrietCheckingStack}
J.~Engelfriet, The power of two-way deterministic checking stack automata, {\em
  Information and Computation} {\bf 80}(2)  (1989)  114--120.

\bibitem{GGH2}
S.~Ginsburg, S.~Greibach and M.~Harrison, One-way stack automata, {\em Journal
  of the ACM} {\bf 14}(2)  (1967)  389--418.

\bibitem{GGH}
S.~Ginsburg, S.~Greibach and M.~Harrison, Stack automata and compiling, {\em
  Journal of the ACM} {\bf 14}(1)  (1967)  172--201.

\bibitem{Greibach}
S.~Greibach, A note on undecidable properties of formal languages, {\em
  Mathematics Systems Theory} {\bf 2}(1)  (1968)  1--6.

\bibitem{CheckingStack}
S.~Greibach, Checking automata and one-way stack languages, {\em Journal of
  Computer and System Sciences} {\bf 3}(2)  (1969)  196--217.

\bibitem{HU}
J.~E. Hopcroft and J.~D. Ullman, {\em Introduction to Automata Theory,
  Languages, and Computation} (Addison-Wesley, Reading, MA, 1979).

\bibitem{Ibarra-JCSS71}
O.~H. Ibarra, Characterizations of some tape and time complexity classes of
  {T}uring machines in terms of multihead and auxiliary stack automata, {\em
  Journal of Computer and System Sciences} {\bf 5}(2)  (1971)  88--117.

\bibitem{Ibarra-JACM1972}
O.~H. Ibarra, A note concerning nondeterministic tape complexities, {\em
  Journal of the ACM} {\bf 19}(4)  (1972)  608--612.

\bibitem{Ibarramultihead}
O.~H. Ibarra, On two-way multihead automata, {\em Journal of Computer and
  System Sciences} {\bf 7}(1)  (1973)  28--36.

\bibitem{DLT2020checking}
O.~H. Ibarra, J.~{Jir\'asek Jr.}, I.~McQuillan and L.~Prigioniero, Space
  complexity of stack automata models, {\em Lecture Notes in Computer
  Science\/},  eds. N.~Jonoska and D.~Savchuk {\em Proceedings of the 24th
  International Conference on Developments in Language Theory, DLT 2020} {\bf
  12086}  (2020), pp. 137--149.

\bibitem{DLT2018checking}
O.~H. Ibarra and I.~McQuillan, Generalizations of checking stack automata:
  Characterizations and hierarchies, {\em Lecture Notes in Computer Science\/},
   eds. M.~Hoshi and S.~Seki {\em Proceedings of the 22nd International
  Conference on Developments in Language Theory, DLT 2018} {\bf 11088}  (2018),
  pp. 416--428.

\bibitem{DLT2017TCS}
O.~H. Ibarra and I.~McQuillan, Variations of checking stack automata: Obtaining
  unexpected decidability properties, {\em Theoretical Computer Science} {\bf
  738}  (2018)  1--12.

\bibitem{KingWrathall}
K.~King and C.~Wrathall, Stack languages and log n space, {\em Journal of
  Computer and System Sciences} {\bf 17}(3)  (1978)  281--299.

\bibitem{Minsky}
M.~L. Minsky, Recursive unsolvability of {P}ost's problem of ``tag'' and other
  topics in theory of {T}uring {M}achines, {\em Annals of Mathematics} {\bf
  74}(3)  (1961)  pp. 437--455.

\bibitem{Monien1}
B.~Monien, Transformational methods and their application to complexity
  problems, {\em Acta Informatica} {\bf 6}  (1976)  95--108.

\bibitem{Monien}
B.~Monien, Two-way multihead automata over a one-letter alphabet, {\em RAIRO.
  Informatique Th\'{e}orique} {\bf 14}(1)  (1980)  67--82.

\bibitem{Seiferas1977}
J.~I. Seiferas, Relating refined space complexity classes, {\em Journal of
  Computer and System Sciences} {\bf 14}(1)  (1977)  100--129.

\bibitem{ShamirBeeri}
E.~Shamir and C.~Beeri, Checking stacks and context-free programmed grammars
  accept {P}-complete languages, {\em Proceedings of the 2nd Colloquium on
  Automata, Languages and Programming\/},  (Springer-Verlag, Berlin,
  Heidelberg, 1974), pp. 27--33.

\bibitem{VogelWagner}
J.~Vogel and K.~Wagner, Two-way automata with more than one storage medium,
  {\em Theoretical Computer Science} {\bf 39}  (1985)  267--280.

\end{thebibliography}
\bibliographystyle{ws-ijfcs}

\end{document}